\documentclass{sig-alternate-05-2015}

\usepackage{array, calc, color, multicol}
\usepackage{graphics, epstopdf, graphicx, epsfig, psfrag, epsf, subfigure}
\usepackage{url}
\newcommand{\ie}{{\em i.e., }}
\newcommand{\Ie}{{\em I.e., }}
\newcommand{\eg}{{\em e.g., }}
\newcommand{\Eg}{{\em E.g., }}

\newcommand{\modelI}{{\tt{Queue I}}}
\newcommand{\modelII}{{\tt{Queue II}}}

\newtheorem{theorem}{Theorem}

\newtheorem{example}{Example}

\DeclareGraphicsExtensions{.eps, .pdf, .png, .jpg}   

\begin{document}

\setcopyright{acmcopyright}

\CopyrightYear{2016}
\setcopyright{acmcopyright}
\conferenceinfo{CarSys'16,}{October 03-07, 2016, New York City, NY, USA} 
\isbn{978-1-4503-4250-6/16/10}\acmPrice{\$15.00} 
\doi{http://dx.doi.org/10.1145/2980100.2980105}

%

\title{Connectivity-Aware Traffic Phase Scheduling for Heterogeneously Connected Vehicles}

\numberofauthors{2} 
%
\author{
\alignauthor Shanyu Zhou\\
       \affaddr{University of Illinois at Chicago}\\
       \email{\tt szhou45@uic.edu}
\alignauthor Hulya Seferoglu\\
       \affaddr{University of Illinois at Chicago}\\
       \email{\tt hulya@uic.edu}
}

\maketitle

\begin{abstract}
We consider a transportation system of heterogeneously connected vehicles, where not all vehicles are able to communicate. Heterogeneous connectivity in transportation systems is coupled to practical constraints such that (i) not all vehicles may be equipped with devices having communication interfaces, (ii) some vehicles may not prefer to communicate due to privacy and security reasons, and (iii) communication links are not perfect and packet losses and delay occur in practice. In this context, it is crucial to develop control algorithms by taking into account the heterogeneity. In this paper, we particularly focus on making traffic phase scheduling decisions. We develop a connectivity-aware traffic phase scheduling algorithm for heterogeneously connected vehicles that increases the intersection efficiency (in terms of the average number of vehicles that are allowed to pass the intersection) by taking into account the heterogeneity. The simulation results show that our algorithm significantly improves the efficiency of intersections as compared to the baselines. 

\end{abstract}


\section{Introduction}\label{sec:intro} 
The increasing population and growing cities introduce several challenges in metropolitan areas, and one of the most challenging areas is transportation systems. In particular, the rapidly increasing number of vehicles in metropolitan transportation systems, has introduced several challenges including higher traffic congestion, delay, accidents, energy consumption, and air pollution. For example, the average of yearly delay per auto commuter due to congestion was 38 hours, and it was as high as 60 hours in large metropolitan areas in 2011 \cite{urban_mobility_report}. The congestion caused 2.9 billion gallons of wasted fuel in 2011, and this figure keeps increasing yearly \cite{urban_mobility_report}, \eg the increase was 3.8\% in Illinois between years 2011 and 2012 \cite{times}. This trend poses a challenge for efficient transportation systems, so new traffic management mechanisms are needed to address the ever increasing transportation challenges.

\begin{figure}[t!]
\begin{center}
\vspace{5pt}
\subfigure[Phase I ($\phi = 1$)]{ \scalebox{.5}{\includegraphics{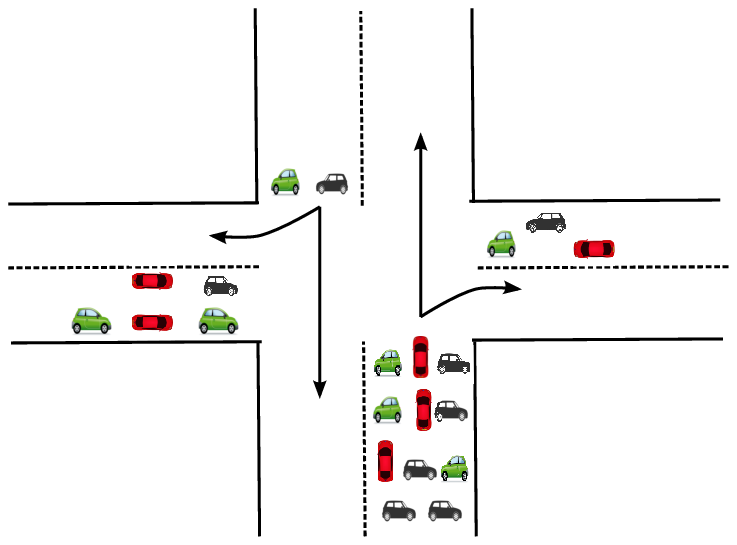}} } 
\subfigure[Phase II ($\phi = 2$)]{ \scalebox{.5}{\includegraphics{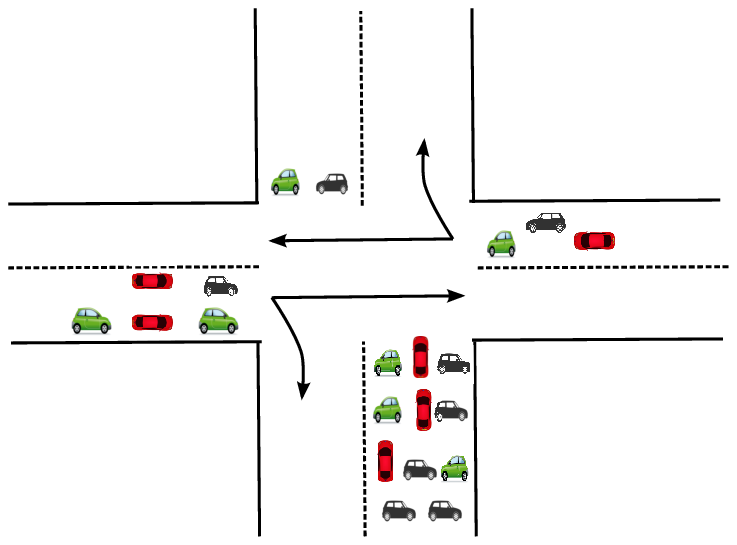}} } 
\subfigure[Phase III ($\phi = 3$)]{ \scalebox{.5}{\includegraphics{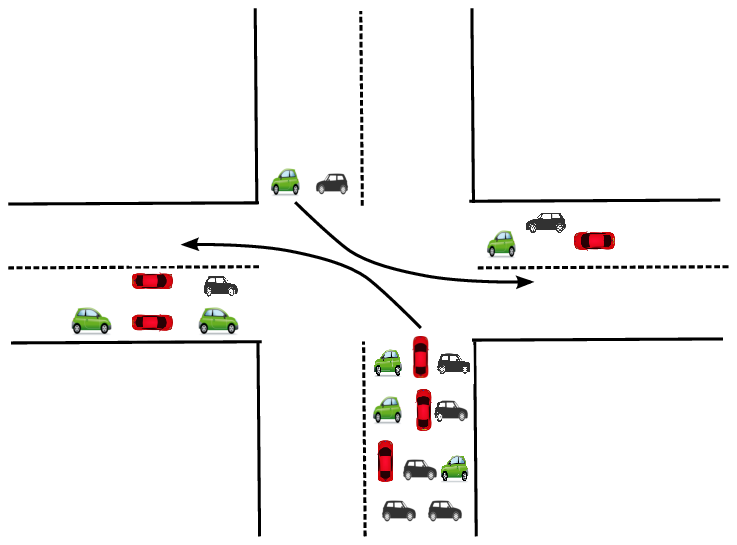}} }
\subfigure[Phase IV ($\phi = 4$)]{ \scalebox{.5}{\includegraphics{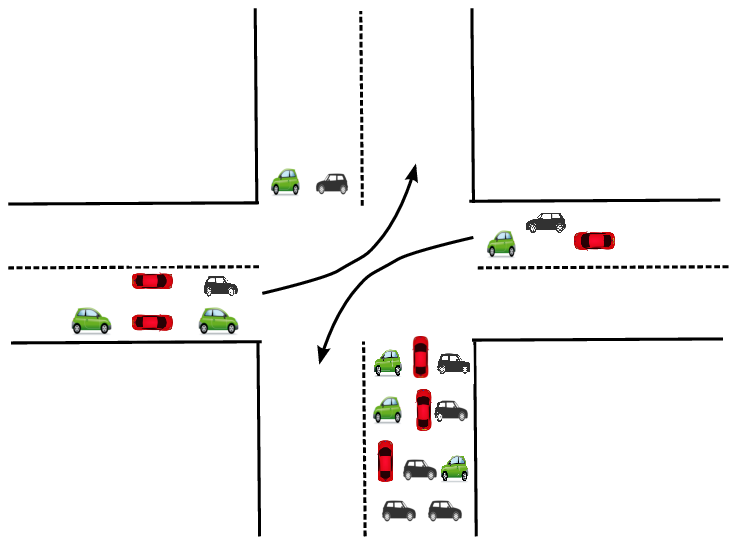}} }
\end{center}
\begin{center}
\vspace{-10pt}
\caption{\label{fig:phases} An example intersection with four possible traffic phases. }
\vspace{-25pt}
\end{center}
\end{figure}

Fortunately, advances in communication and networking theories offer vast amount of opportunities to address ever increasing challenges in transportation systems. In particular, connected vehicles, \ie vehicles that are connected to the Internet via cellular connections and to each other via device-to-device (D2D) connections such as Bluetooth or WiFi-Direct \cite{aboudolas}, are able to transmit and receive information to improve the control and management of traffic, which has potential of reducing congestion, delay, energy, and improving reliability. In this context, it is crucial to understand how heterogeneous communication affects the performance of transportation systems.

Heterogeneity in transportation systems is coupled to practical constraints such that (i) not all vehicles may be equipped with devices having communication interfaces, (ii) some vehicles may not prefer to communicate due to privacy and security reasons, and (iii) communication links are not perfect and packet losses and delay occur in practice. It is crucial to develop control algorithms by taking into account the heterogeneity. In this paper, we particularly focus on making traffic phase scheduling decisions. The next two examples illustrate the traffic phase scheduling problem and the impact of heterogeneous communications on the scheduling. 

\vspace{-5pt}
\begin{example} \label{ex:sch_example}
Let us consider Fig.~\ref{fig:phases}, which shows an isolated intersection, and all four possible traffic light phases. Traffic lights could be configured in four different phases: Phases I, II, III, and IV. \Eg Phase I corresponds to the case that only north-south and south-north bounds are allowed to pass through the intersection. The traffic light scheduling determines the phase that should be activated. Note that only one phase could be activated at a time. It is clear that scheduling decisions should be made based on the congestion levels of different directions (or traffic bounds). For example, selecting either Phase I or Phase III in the specific example of Fig.~\ref{fig:phases} looks a better decision as compared to Phase II or Phase IV, because Phase I and Phase III have a larger number of vehicles in their corresponding queues. 
\end{example}

Example~\ref{ex:sch_example} is a widely known problem in network control and optimization theory, and the optimal solution to this problem is the popular max-weight algorithm \cite{tassiulas}. The broader idea behind max-weight algorithm is to prioritize the scheduling decisions with larger weights, which corresponds to congestion level, loss probabilities, and link qualities. The max-weight idea is applied to transportation systems as well in previous work \cite{gregoire, varaiya_13, varaiya, wongpir} that schedules traffic phases according to congestion levels, which has potential of allowing more vehicles to pass and reduce waiting times at intersections. This approach works well in a scenario that the directions of all vehicles are known a-priori. For example, if all devices communicate with the traffic light in terms of their intentions about their directions (\eg turn right, go straight, etc.), the traffic light determines which phase to activate using the max-weight scheduling algorithm. However, due to heterogeneity of communication in connected vehicles, only a percentage of vehicles communicate their intentions. In this heterogeneous setup, new connectivity-aware traffic phase scheduling algorithms are needed as illustrated in the next example. 

\begin{figure}[t!]
\begin{center}
\subfigure[Only the first vehicle communicates]{ \scalebox{1.0}{\includegraphics{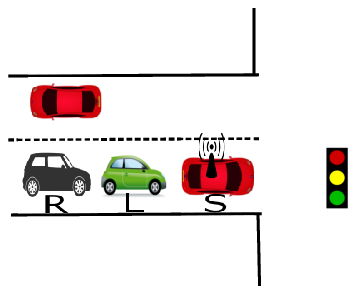}} }
\subfigure[Only the second vehicle communicates]{ \scalebox{1.0}{\includegraphics{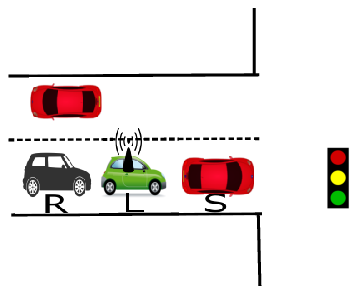}} }
\end{center}
\begin{center}
\vspace{-15pt}
\caption{\label{fig:intro_exampleFig} An example single-lane intersection, where vehicles are going straight, turning left and turning right respectively.}
\vspace{-30pt}
\end{center}
\end{figure}

\vspace{-5pt}
\begin{example} \label{ex:sch_example_cont}
Let us consider Fig. \ref{fig:intro_exampleFig}, which shows one of the four incoming traffic lanes in an intersection. This is a one-way single-lane road, where we call the first vehicle at the intersection as the head-of-line (HoL) vehicle. In Fig. \ref{fig:intro_exampleFig}(a), the HoL vehicle has communication ability, and the vehicles are going straight, turning left, and turning right, respectively. In this case, the traffic light knows that the HoL vehicle is going straight (because the HoL vehicle communicates), so it arranges its phase accordingly. 

Now let us consider Fig. \ref{fig:intro_exampleFig}(b), where the directions of vehicles are the same; \ie straight, left, and right. Yet, in this scenario HoL vehicle does not communicate, but only the vehicle behind HoL communicates. In this case, although the traffic light knows that the second vehicle is going to the left, it has no idea of the HoL vehicle's intention. If the traffic phase, possibly determined as a solution to the max-weight algorithm, does not match the intention of the HoL vehicle, then the HoL vehicle blocks the other vehicles at the intersection, and no vehicles can pass. Similarly, HoL blocking can be observed in more involved multiple-lane scenarios \cite{allerton15_tech}. As seen, the max-weight algorithm may not be optimal in some scenarios due to heterogeneous connectivity, which makes the development of new scheduling algorithms, by taking into account heterogeneity, crucial. 
\hfill $\Box$
\end{example}

In this paper, we develop a connectivity-aware traffic phase scheduling algorithm by taking into account heterogeneous communications of connected vehicles. Our approach follows a similar idea to the max-weight scheduling algorithm, which makes scheduling decisions based on congestion levels at intersections. However, our algorithm, which we name Connectivity-Aware Max-Weight (CAMW), is fundamentally different from the max-weight as we take into account heterogeneous communications while determining congestion levels. In particular, CAMW has two critical components to determine congestion: (i) Expectation: This component calculates the expected number of vehicles that can pass through the intersection at every phase based on the number of vehicles, and the percentage of communicating vehicles at the intersection. (ii) Learning: This component learns the directions of vehicles even if the vehicles do not directly communicate with the traffic light. The expectation  and learning components of our algorithm operate together in harmony to make better decision on traffic phase scheduling. The simulation results demonstrate that CAMW algorithm significantly improves the intersection efficiency (in terms of the average number of vehicles that are allowed to pass the intersection) over the baseline algorithm; max-weight. The following are the key contributions of this work:

\begin{itemize}
\item We investigate the impact of heterogeneous communication on traffic phase scheduling problem in transportation networks. 
We develop a connectivity-aware traffic scheduling algorithm, which we name Connectivity-Aware Max-Weight (CAMW), by taking into account the congestion levels at intersections and the heterogeneous communications.
\item The crucial parts of CAMW are expectation and learning components. In the expectation component, we characterize the expected number of vehicles that can pass through the intersections by taking into account the heterogeneous connectivity. In the learning component, we infer the directions of vehicles even if they do not directly communicate. The expectation and learning components collectively determine the number of vehicles that can pass through the intersections.
\item We evaluate CAMW via simulations, which confirm our analysis, and show that our algorithm significantly improves intersection efficiency as compared to the baseline; the max-weight algorithm.
\end{itemize}

The structure of the rest of this paper is as follows. Section \ref{sec:related} presents the related work. Section \ref{sec:system model} introduces the system model. Section \ref{sec:alg_memo} develops our connectivity-aware traffic phase scheduling algorithm by taking into account heterogeneous communications. Section \ref{sec:sims} presents the simulation results. Section \ref{sec:conclusion} concludes the paper.

\section{Related work}\label{sec:related}

This work combines ideas from traffic phase scheduling, queuing theory, and network optimization. In this section, we discuss the most relevant literature from these areas.

{\em Traffic phase scheduling:} Design and development of traffic phase scheduling algorithms have a long history; more than 50 years \cite{miller}. Thus, there is huge literature in the area, especially on the design of optimal pre-timed policies \cite{miller, gartner, cascetta}, which activate traffic phases according to a time-periodic pre-defined schedule. These policies do not meet expectations under changing arrival times, which require adaptive control \cite{pitu_ning}. The adaptive control mechanisms including \cite{diakaki}, \cite{gartner}, \cite{henry}, \cite{hunt}, \cite{lowrie} and \cite{mirchandani}, optimize control variables, such as traffic phases, based on traffic measures, and apply them on short term. 

{\em Queueing theory:} Using queuing theory to analyze transportation systems has also very long history  \cite{webster}. \Eg \cite{miller, newell, dirk} considered one-lane queues and calculated the expected queue length and arrivals using probability generation functions. Other modeling strategies are also studied; such as the queuing network model \cite{osorio}, cell transmission model \cite{lo}, store-and-forward \cite{aboudolas}, and petri-nets \cite{febbraro}. 

{\em Network optimization and its applications to transportation systems:} Max-weight scheduling algorithm and backpressure routing and scheduling algorithms \cite{tassiulas} arising from network optimization area has triggered significant research in wireless networks \cite{neely, neely_modiano}. This topic has also inspired research in transportation systems \cite{gregoire, varaiya_13, varaiya, wongpir}. Feedback control algorithms to ensure maximum stability are proposed both under deterministic arrivals \cite{varaiya} and stochastic arrivals \cite{varaiya_09, wongpir} following backpressure idea. The infinite buffer assumption of backpressure framework is studied by capacity aware back-pressure algorithm in \cite{jean}.

{\em Our work in perspective:} As compared to the previous work briefly summarized above, our work focuses on connected vehicles and investigates the scenario where vehicles communicate heterogeneously. In this scenario, we develop an efficient connectivity-aware traffic phase scheduling algorithm by employing expectation and learning of congestion levels at intersections. 

Our previous work \cite{allerton15_tech} investigates the impact of the blocking problem at intersections, characterizes the waiting times, and develops a shortest delay routing algorithm in transportation systems. As compared to this work, in this paper, we develop a connectivity-aware traffic phase scheduling algorithm by taking into account heterogeneous communications.

\vspace{-5pt}
\section{System model}\label{sec:system model}
In this section, we present our system model including traffic lights and phases as well as our queuing models of the traffic. 

{\em Traffic lights and phases:} 
In our system model, we focus on an intersection controlled by a traffic light. The four traffic phases we consider in this paper are shown in Fig.~\ref{fig:phases}. We define $\phi$ as a phase decision, \eg $\phi = 1$ corresponds to Phase I in Fig.~\ref{fig:phases}. The set of phases is $\Phi$, and $\phi \in \Phi$. 

We consider that time is slotted, and at each time slot $t$, a phase decision is made. Each traffic phase lasts for $n$ time slots. 
Vehicles have a chance to pass the intersection only when the corresponding traffic phase is active, \ie ON. For instance, vehicles in the south-north bound lanes may pass the intersection only when phase $\phi = 1$ is ON in Fig.~\ref{fig:phases}. 

{\em Modeling intersections with queues:}
We model the isolated intersection as a set of queues following \cite{allerton15_tech}. Typically, there are four queues for each direction (for south-north, north-south, west-east, and east-west bound) at an intersection. We specifically focus on one direction and model it using two models: \modelI, which is one-lane model shown in Fig.~\ref{queuemodels}(a) and \modelII; which is a one+two lane model shown in Fig.~\ref{queuemodels}(b).

\begin{figure}[t!]
\begin{center}
\subfigure[\modelI]{ \scalebox{0.8}{\includegraphics{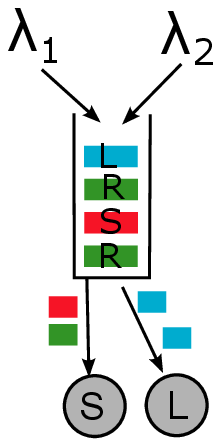}} }\hspace{40pt}
\subfigure[\modelII]{ \scalebox{0.8}{\includegraphics{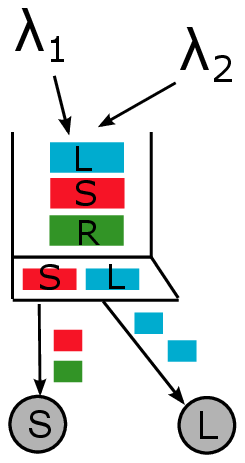}} }
\end{center}
\begin{center}
\vspace{-10pt}
\caption{\label{queuemodels} Two queuing models considered in this paper, where $\lambda_1$ and $\lambda_2$ are the arrival rates of straight-going and left-turning traffic, respectively.(a) Single-lane traffic model. (b) One+two lane model.}
\vspace{-35pt}
\end{center}
\end{figure}

Note that for both of \modelI~ and \modelII, we can consider straight-continuing and right-turning traffic as the same traffic, since they share the similar right of way. Thus, to demonstrate the analysis in a simple way, we simply consider that the right-turning and straight-continuing traffics are combined together, and we call both right-turning and straight-continuing vehicles as straight-going vehicles.

At each slot, vehicles arrive into intersections, where $\lambda_1$ and $\lambda_2$ are the average arrival rates of straight-going and left-turning vehicles, respectively. In our analysis, the arrivals can follow any i.i.d. distribution. In this setup, when a vehicle enters the intersection, it can connect to the traffic light either using cellular or vehicle-to-vehicle communications. Thus, it can communicate its intention with the traffic light about its destination, \ie turning left, going straight, etc. The probability of communication for each vehicle is $\rho$. 

If a vehicle does not communicate, we model their intentions probabilistically, where $p_1$ is the probability that a vehicle (which does not communicate its intention) will go straight, while $p_2$ is the probability that it will turn left.  



\vspace{-5pt}
\section{Connectivity-Aware Traffic Phase Scheduling}\label{sec:alg_memo} 

\subsection{CAMW: Connectivity-Aware Max-Weight}
In this section, we develop our connectivity-aware traffic phase scheduling algorithm by taking into account heterogeneous communications. We consider the setup shown in Fig. \ref{fig:phases} for phases. Our scheduling algorithm, which we call Connectivity-Aware Max-Weight (CAMW), determines the phase $\phi$ by optimizing
\begin{align}\label{eq:sche_obj_memo}
\max_{\boldsymbol \phi} \mbox{ } &  \sum_{i \in \{1, \ldots 4\}} Q_i(t)\tilde{E}(K^{\phi}_i(t)) \nonumber \\
\mbox{s.t.} \mbox{ }  & \phi \in \Phi. 
\end{align} where $Q_i(t)$ is the number of vehicles in the $i$th incoming queue at time slot $t$, and $\tilde{E}(K^{\phi}_i(t))$ is the estimated number of vehicles that can pass the intersection from the $i$th incoming queue under traffic phase $\phi \in \Phi$. Note that one active phase lasts for $n$ time slots and it takes one time slot for a vehicle to pass the intersection. In other words, at most $n$ vehicles in a queue can pass the intersection during one green light phase. The optimization problem in (\ref{eq:sche_obj_memo}) applies to all queuing models (\ie includes both \modelI~ and \modelII~). 

Note that (\ref{eq:sche_obj_memo}) determines the phase by taking into account $Q_i(t)$ and $\tilde{E}(K^{\phi}_i(t))$. The queue size information $Q_i(t)$ can be easily determined by traffic lights using sensors that count the number of approaching vehicles. In other words, (\ref{eq:sche_obj_memo}) prioritizes phases with larger $Q_i(t)$ values. This is an approach followed by the classical max-weight algorithm. However, as we discussed earlier, using $Q_i(t)$ alone is not sufficient when vehicles heterogeneously communicate with traffic lights. In this case, since each device has different destinations, blocking can occur. \Ie even if $Q_i(t)$ is large, the number of vehicles that can pass through the intersection could be small due to blocking. Thus, to reflect this fact, we include the term $\tilde{E}(K^{\phi}_i(t))$ in the optimization problem. 

$\tilde{E}(K^{\phi}_i(t))$ is the estimated number of vehicles that can pass the intersection from the $i$th incoming queue under traffic phase $\phi \in \Phi$. $\tilde{E}(K^{\phi}_i(t))$ is found using two steps: expectation and learning. The key idea behind expectation part is to calculate the expected number of vehicles, which is ${E}(K^{\phi}_i(t))$, that can pass the intersection at phase $\phi$, while the key idea of the learning part is to fine tune ${E}(K^{\phi}_i(t))$ and find  $\tilde{E}(K^{\phi}_i(t))$  by learning the directions of vehicles that do not communicate. In the next two sections, we present the expectation and learning components of CAMW.

\begin{figure} [t!]
\centering
\scalebox{1.0}{\includegraphics{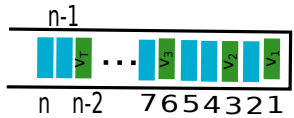}}
\vspace{-5pt}
\caption{\label{fig:communicatingvehiclesinthequeue} An illustrative example of communicating vehicles in a queue at a time slot. Communicating vehicles are at labeled locations; $v_1, v_2,\cdots, v_T$.}
\vspace{-15pt}
\end{figure}

\subsection{Expectation}\label{sec:expectation}
\subsubsection{Calculation of $E(K^{\phi}_i(t))$ for \modelI}\label{memo_m1_analysis}
Let us focus on phase $\phi \in \Phi$ and the $i$th queue, where $i \in \{1,2,3,4\}$. In this setup, $T(t)$ ($T(t)\leq n$) denotes the number of vehicles that have communication abilities at time slot $t$, and $v_l(t)$ ($l=1,2,\cdots,T$) denotes the location of the $l$th communicating vehicle in the queue. For example, $v_2(t)=3$ means that the second communicating vehicle in the queue is actually the third vehicle in the queue.  Fig.~\ref{fig:communicatingvehiclesinthequeue} illustrates an example locations of communicating vehicles. Note that the vehicles that do not communicate are not assigned any location labels. 

Now, let us define two conditions; $C_1$ and $C_2$. The first condition $C_1$ requires that all communicating vehicles would like to go to the same direction and aligned with the traffic phase, while the second condition $C_2$  corresponds to the case that the first communicating vehicle that is not aligned with the traffic phase is in the location of $v_L(t)$ ($L=1,2,\cdots,T$). Note that the conditions $C_1$ and $C_2$ are complementary. The next theorem characterizes the expected number of vehicles that would leave queue $i$ at phase $\phi$. 

\begin{theorem} 
Assume that all the queues  in an intersection follow \modelI. The expected number of vehicles that would leave the $i$th  queue and pass the intersection at traffic phase $\phi \in \Phi$ is characterized by
\begin{align}\label{eqn:final_compact_k(t)}
E(K^{\phi}_i(t)) = \begin{cases} \sum_{l=0}^{T(t)}\frac{p^{1-l}_1}{p_2}((p_1+p_2v_l(t))p^{v_l(t)-1}_1\nonumber\\
 + (1-2p_1-p_2v_{l+1}(t))p^{v_{l+1}(t)-2}_1)\nonumber\\
+np^{n-T(t)}_1, & \mbox{    {\em if} $C_1$  {\em holds}  }\\
\\
\sum_{l=0}^{L-1}\frac{p^{1-l}_1}{p_2}((p_1+p_2v_l(t))p^{v_l(t)-1}_1\nonumber \\
 + (1-2p_1-p_2v_{l+1}(t))p^{v_{l+1}(t)-2}_1)\nonumber \\
 +(v_L(t)-1)p^{v_L(t)-L}_1,   & \mbox{    {\em if} $C_2$  {\em holds.}  }
\end{cases} 
\\
\end{align}
\end{theorem}

\begin{proof}
The proof is provided in Appendix \ref{proof:k(t)_m1}. \hfill 
\end{proof}

\vspace{-5pt}
\subsubsection{Calculation of $E(K^{\phi}_i(t))$ for \modelII \label{memo_m2_analysis} }
\modelII~ assumes that there are dedicated lanes for left-turning and straight-going vehicles, which makes it fundamentally different than \modelI. In this setup, we consider that traffic lights can sense whether the HoL location of each dedicated lane is empty or not. Thus, in \modelII, the first two vehicles in the queue will indirectly communicate their intentions to the traffic light. Fig. \ref{fig:holconfiguration} demonstrates four possible configurations for HoL vehicles. For example, in Fig. \ref{fig:holconfiguration}(a), HoL position of the straight going lane is empty (shown with {\em E}), the traffic light will know that two vehicles in the queue will turn left. On the other hand, in Fig. \ref{fig:holconfiguration}(b), the traffic light knows that in the dedicated lanes, one vehicle will go straight, and the other will turn left, but it does not know the intentions of the other vehicles as long as they do not explicitly communicate with the traffic light. 

\begin{figure}[t!]
\begin{center}
\subfigure[Conf. I]{ \scalebox{1.0}{\includegraphics{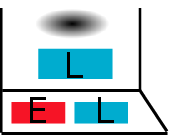}} } \hspace{20pt}
\subfigure[Conf. II]{ \scalebox{1.0}{\includegraphics{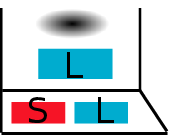}} } \\
\subfigure[Conf. III]{ \scalebox{1.0}{\includegraphics{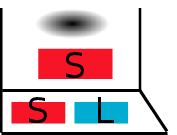}} }  \hspace{20pt}
\subfigure[Conf. IV]{ \scalebox{1.0}{\includegraphics{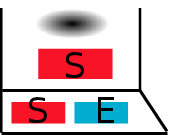}} }
\end{center}
\begin{center}
\vspace{-10pt}
\caption{\label{fig:holconfiguration} Four possible configurations (Conf. I to Conf IV) for the first three vehicles in \modelII, where $L$ and $S$ denote that the intention of the vehicle is to turn left or go straight, respectively, while $E$ denotes that the location is empty (due to previous blocking).}
\vspace{-25pt}
\end{center}
\end{figure}

The crucial observation with \modelII~ is that if the vehicles that indirectly communicate with the traffic light are separated from the queue, the rest of the vehicles form a {\em sub-queue}. For example, all the vehicles other than (i) the first two left-turning vehicles in Fig.~\ref{fig:holconfiguration}(a), and (ii) the two vehicles that are going straight and turning left  in Fig.~\ref{fig:holconfiguration}(b), form  a {\em sub-queue}. The important property of the {\em sub-queue} is that it follows \modelI, and can be modeled using the location labels as shown in Fig. \ref{fig:communicatingvehiclesinthequeue}. Thus, we can calculate $E(K^{\phi}_i(t))$ of \modelII~ using the similar analysis we have in Section~\ref{memo_m1_analysis}. Next, we provide the details of our $E(K^{\phi}_i(t))$ calculation. 

Let $T(t)$ denotes the number of communicating vehicles in the {\em sub-queue} at time $t$, $C_3$ is the condition that all communicating vehicles in the {\em sub-queue} go to the same direction aligned with the traffic phase, and $C_4$ denotes the condition that the first communicating vehicle in the {\em sub-queue} that goes to a different direction than what the traffic phase allows is at location $v_L(t)$ ($L=1,2,\cdots,T$). The next theorem characterizes the expected number of vehicles that would leave queue $i$ at phase $\phi$ for model \modelII. 

\begin{theorem}
Assume that all the queues  in an intersection follow \modelII. Then, if the first three vehicles of the $i$th incoming queue are in the form of Fig. \ref{fig:holconfiguration}(a) or Fig. \ref{fig:holconfiguration}(d), the expected number of transmittable vehicles is characterized by 
\begin{align} \label{eqn:final_compact_k(t)_m2}
E(K^{\phi}_i(t)) = \begin{cases} 2+\sum_{l=0}^{T(t)}\frac{p^{1-l}_1}{p_2}((p_1+p_2v_l(t))\nonumber \\
p^{v_l(t)-1}_1+ (1-2p_1-p_2v_{l+1}(t)) \nonumber \\
p^{v_{l+1}(t)-2}_1)+(n-2)p^{n-2-T(t)}_1, & \mbox{ {\em if} $C_3$ {\em holds} } \\
\\
2+\sum_{l=0}^{L-1}\frac{p^{1-l}_1}{p_2}((p_1+p_2v_l(t))\nonumber \\
p^{v_l(t)-1}_1+ (1-2p_1-p_2v_{l+1}(t)) \nonumber \\
p^{v_{l+1}(t)-2}_1)+(v_L(t)-1)p^{v_L(t)-L}_1, & \mbox{{\em if} $C_4$ {\em holds} } 
\end{cases}  
\\
\end{align} where $T(t) \leq n-2$.

And if the first three vehicles of the $i$th incoming queue are in the form of Fig. \ref{fig:holconfiguration}(b) or Fig. \ref{fig:holconfiguration}(c), the expected number of transmittable vehicles is characterized by
\begin{align} \label{eqn:final_compact_k(t)_m2_2}
E(K^{\phi}_i(t)) = \begin{cases} 1+\sum_{l=0}^{T(t)}\frac{p^{1-l}_1}{p_2}((p_1+p_2v_l(t)) \nonumber \\
p^{v_l(t)-1}_1+ (1-2p_1-p_2v_{l+1}(t)) \nonumber \\
p^{v_{l+1}(t)-2}_1)+(n-1)p^{n-1-T(t)}_1, & \mbox{{\em if} $C_3$ {\em holds}}\\
\\
1+\sum_{l=0}^{L-1}\frac{p^{1-l}_1}{p_2}((p_1+p_2v_l(t))\nonumber \\
p^{v_l(t)-1}_1+ (1-2p_1-p_2v_{l+1}(t))\nonumber \\
p^{v_{l+1}(t)-2}_1)+(v_L(t)-1)p^{v_L(t)-L}_1, & \mbox{{\em if} $C_4$ {\em holds}}
\end{cases} 
\\
\end{align} 
where $T(t) \leq n-1$.
\end{theorem}

\begin{proof}
The number of vehicles that can be guaranteed to pass the intersection under certain traffic phase depends on the configuration of the first three vehicles in the queue. First, we consider the case that the first three vehicles are in the form of Fig. \ref{fig:holconfiguration}(a) or Fig. \ref{fig:holconfiguration}(d). In this case, at least two vehicles can pass the intersection for the corresponding traffic phase, so we need to consider the rest of the vehicles, \ie $n-2$ vehicles assuming that $n$ is the queue size. Noting that $n-2$ vehicles form a {\em sub-queue} in this setup, and assuming that  $T(t)$ ($T(t)\leq n-2$) vehicles communicate the {\em sub-queue}, it is clear that the {\em sub-queue} is represented by \modelI. Thus,  (\ref{eqn:final_compact_k(t)_m2}) is obtained by adding two to (\ref{eqn:final_compact_k(t)}). 

On the other hand, if the first three vehicles are in the form of Fig. \ref{fig:holconfiguration}(b) or Fig. \ref{fig:holconfiguration}(c), at least one vehicle can pass the intersection at any traffic phase configuration. In this scenario, one vehicle is considered as guaranteed to be transmitted, and the rest of the vehicles ($n-1$ vehicles) form a {\em sub-queue}. Similar to above discussion, the {\em sub-queue} follows \modelI, so (\ref{eqn:final_compact_k(t)_m2_2}) is obtained by adding one to  (\ref{eqn:final_compact_k(t)}). This concludes the proof.  \hfill
\end{proof}

~~~

\subsection{Learning}\label{sec:learn}
In the previous section, we characterized the expected number of vehicles $E(K^{\phi}_i(t))$ that can pass an intersection at phase $\phi$ from queue $i$. However, in our CAMW algorithm, which solves (\ref{eq:sche_obj_memo}), we do not use $E(K^{\phi}_i(t))$. The reason is that $E(K^{\phi}_i(t))$ is an expected value and its granularity is poor. In other words, if we use $E(K^{\phi}_i(t))$ in (\ref{eq:sche_obj_memo}), we may end up with choosing a traffic phase that allows no vehicles passing the intersection. In this case, the intersection is {\em blocked}. More importantly, once the intersection is {\em blocked}, if we keep using  $E(K^{\phi}_i(t))$ in (\ref{eq:sche_obj_memo}), we may end up with choosing the wrong traffic phase next time with high probability, which leads to a deadlock. To address this issue, we introduce the learning mechanism, which works in the following way.

We assume that traffic lights can infer if blocking occurs at intersections, and use this information in future decisions. For example, assume that  the selected traffic phase at time $t-1$ is $\phi = 1$ (as shown in Fig. \ref{fig:phases}(a)), and $\tilde{E}(K^{\phi=1}_i(t-1)) = E(K^{\phi=1}_i(t-1))$. If blocking occurs, then the traffic light can learn that both of the HoL vehicles in south-north bound queues must go left. Using this information, $\tilde{E}(K^{\phi=1}_i(t))$ is set to zero  at time $t$ so that $\phi=1$  is not selected again. $\tilde{E}(K^{\phi=1}_i(t+\Delta))$ is set to ${E}(K^{\phi=1}_i(t+\Delta))$ again immediately after some vehicles are transmitted from the queues. This may take $\Delta$ time slots. This learning mechanism applies to both \modelI~ and \modelII, but in \modelII, separate lanes for each direction makes the learning process by default. \Ie in \modelII, $\tilde{E}(K^{\phi}_i(t)) = E(K^{\phi}_i(t))$, $\forall t$. 

%
%
%
%


\vspace{-5pt}
\section{Performance Evaluation}\label{sec:sims}
In this section, we consider an intersection controlled by a traffic light. Each arriving vehicle to the intersection can communicate with probability $\rho$. Each green phase lasts for one or more time slots. We assume that the arrival rate to each queue in the intersection is the same; \ie $\lambda_1$ and $\lambda_2$ are the same $\forall i \in \{1,2,3,4\}$. We present the simulation results of our Connectivity-Aware Max-Weight (CAMW) algorithm for both of \modelI~ and \modelII, as compared to the baseline, the max-weight algorithm, which is briefly described next.

\vspace{-5pt}
\subsection{The baseline: max-weight algorithm}
The max-weight scheduling algorithm determines a traffic phase as a solution to
\begin{align}\label{eq:sche_obj_max_weight} 
\max_{\boldsymbol \rho} \mbox{ } &  \sum_{i=1}^{4} Q_i(t)K^{\phi}_i(t) \nonumber \\
\mbox{s.t.} \mbox{ }  & \phi \in \Phi,
\end{align} where $K^{\phi}_i(t)$  is the weight of queue $i$ for phase $\phi$.\footnote{Note that $K^{\phi}_i(t) = 1$ in the original max-weight algorithm, while it varies in (\ref{eq:sche_obj_max_weight}) as explained in this section. Thus, although we call this baseline the max-weight algorithm, it is actually the improved version of the classical max-weight algorithm.} The value of $K^{\phi}_i(t)$ depends on the intersection type and the corresponding queuing models, which is explained next.

First, let us consider \modelI. If the HoL vehicle in the $i$th queue can communicate, then $K^{\phi}_i(t)=1$ for the phase that is aligned with the direction of the HoL vehicle and $K^{\phi}_i(t)=0$ for the other three phases. If the HoL vehicle cannot communicate, max-weight considers $K^{\phi}_i(t)=1$ for the phases that control the $i$th queue if the queue length is larger than zero. 

Second, we assume that all the queues in the intersection follow \modelII. In this setup, we take into account the first two vehicles in the dedicated lanes. For example, if the first two vehicles from the $i$th incoming queue are in the form of Fig. \ref{fig:holconfiguration}(a), then $K^{\phi}_i(t)=1$ for the left turning phase, and $K^{\phi}_i(t)=0$ for the other phases. On the other hand, if the first two vehicles are in the form of Fig. \ref{fig:holconfiguration}(b), then $K^{\phi}_i(t)=1$ for both the left-turning and straight-going phases.

\subsection{Evaluation of CAMW for \modelI}
\begin{figure}[!t]
\begin{center}
\subfigure[Max-weight]{ \scalebox{.243}{\includegraphics{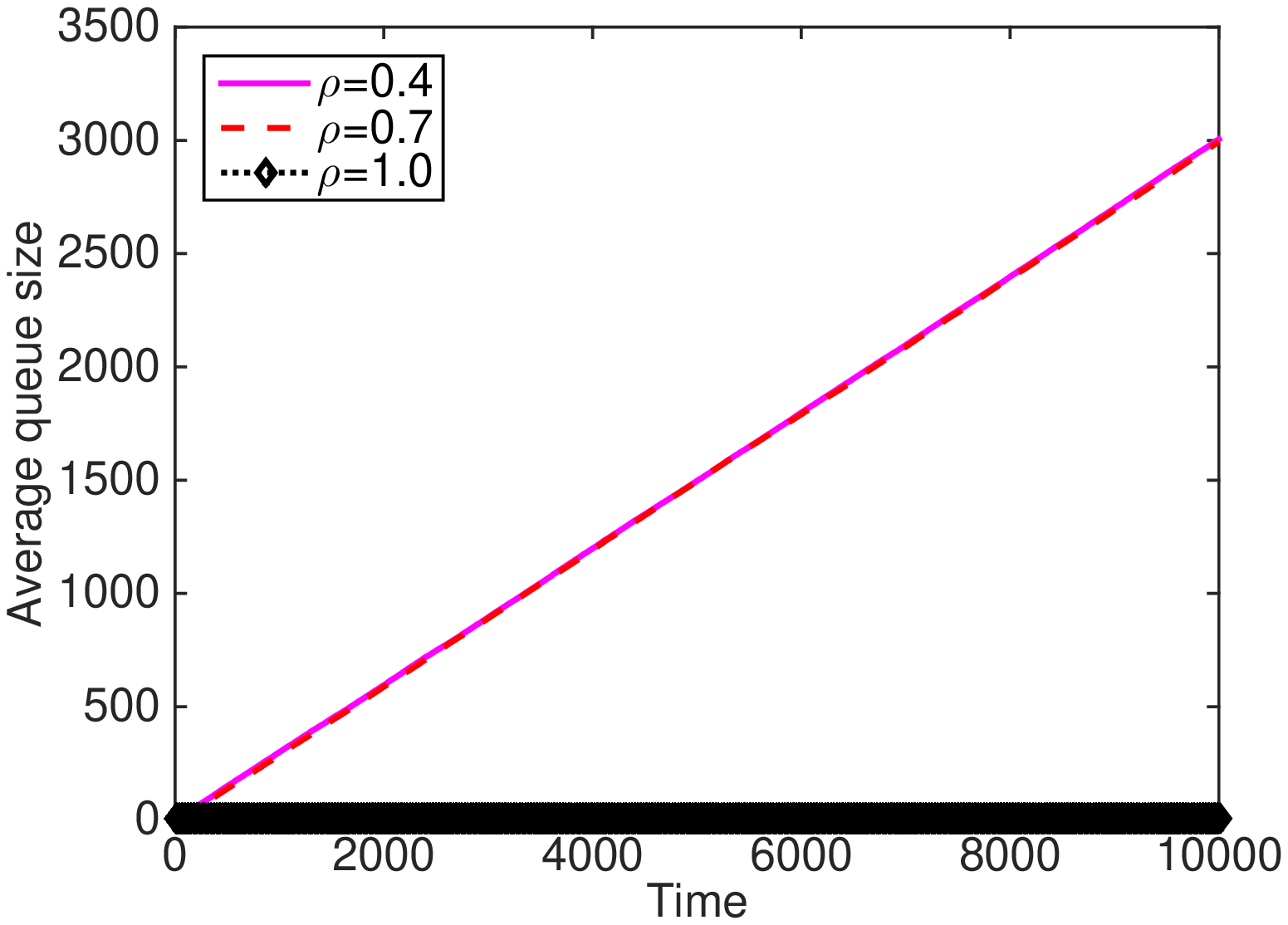}} } 
\subfigure[CAMW]{ \scalebox{.28}{\includegraphics{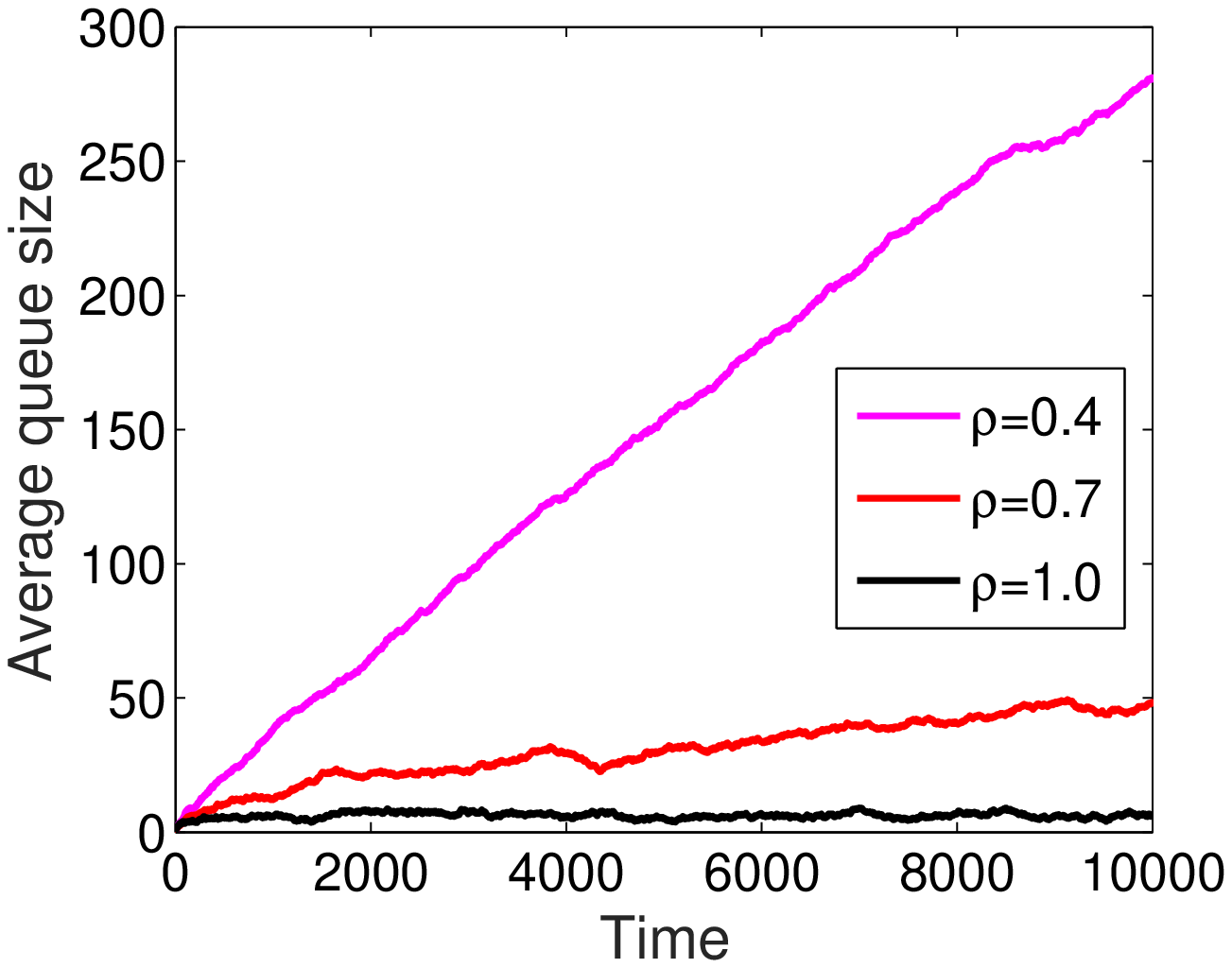}} } \\
\end{center}
\begin{center}
\vspace{-20pt}
\caption{\label{fig:queuesize_time} The average queue size versus time for \modelI~. Each green phase lasts for two time slots. The arrival rate is $\lambda_1=0.18$ and $\lambda_2=0.12$ to each of the queue in the intersection.}
\vspace{-35pt}
\end{center}
\end{figure}

We first assume all the queues in the intersection follow \modelI, and evaluate our CAMW algorithm as compared to the baseline; max-weight. The evolution of the average queue size of the intersection for different scheduling algorithms is presented in Fig. \ref{fig:queuesize_time}. Each green phase lasts for two time slots. The arrival rate is $\lambda_1=0.18$ and $\lambda_2=0.12$ to each of the queue in the intersection. It can be observed that when the communication probability is $\rho=1.0$, both of the algorithms have the similar performance. This is because every vehicle can communicate, so the max-weight algorithm, since the traffic light can communicate with the HoL vehicle, can align the phases with the direction of HoL vehicle. However, when the communication probability reduces to $\rho = 0.7$, max-weight cannot stabilize the queues, while CAMW stabilizes. When $\rho = 0.4$, neither CAMW nor max-weight can stabilize the queues, because the arrival rates fall out of the stability region. As can be seen CAMW supports higher traffic rates than the max-weight algorithm thanks to exploiting connectivity of vehicles. 

\begin{figure}[!t]
\begin{center}
\subfigure[$\rho=0.1$]{ \scalebox{.30}{\includegraphics{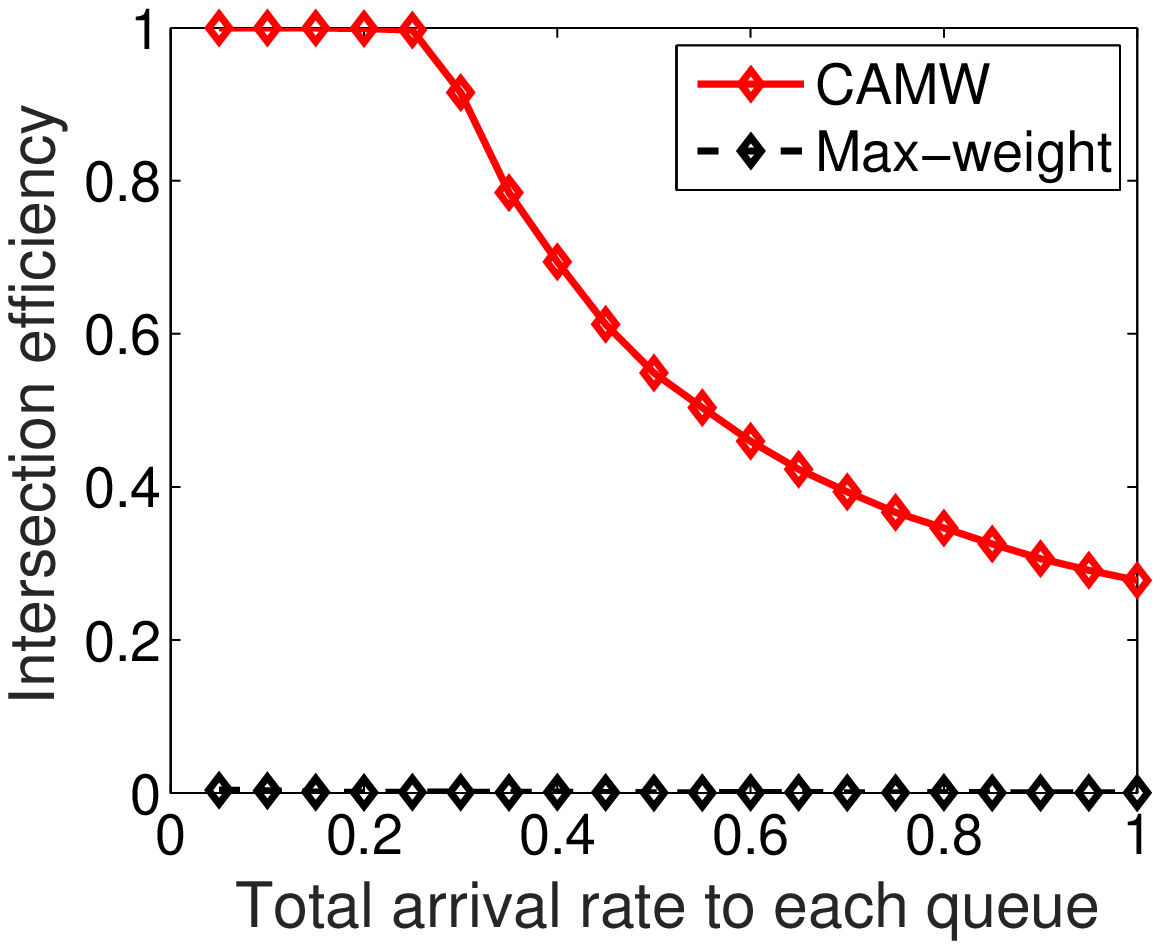}} } 
\subfigure[$\rho=0.4$]{ \scalebox{.30}{\includegraphics{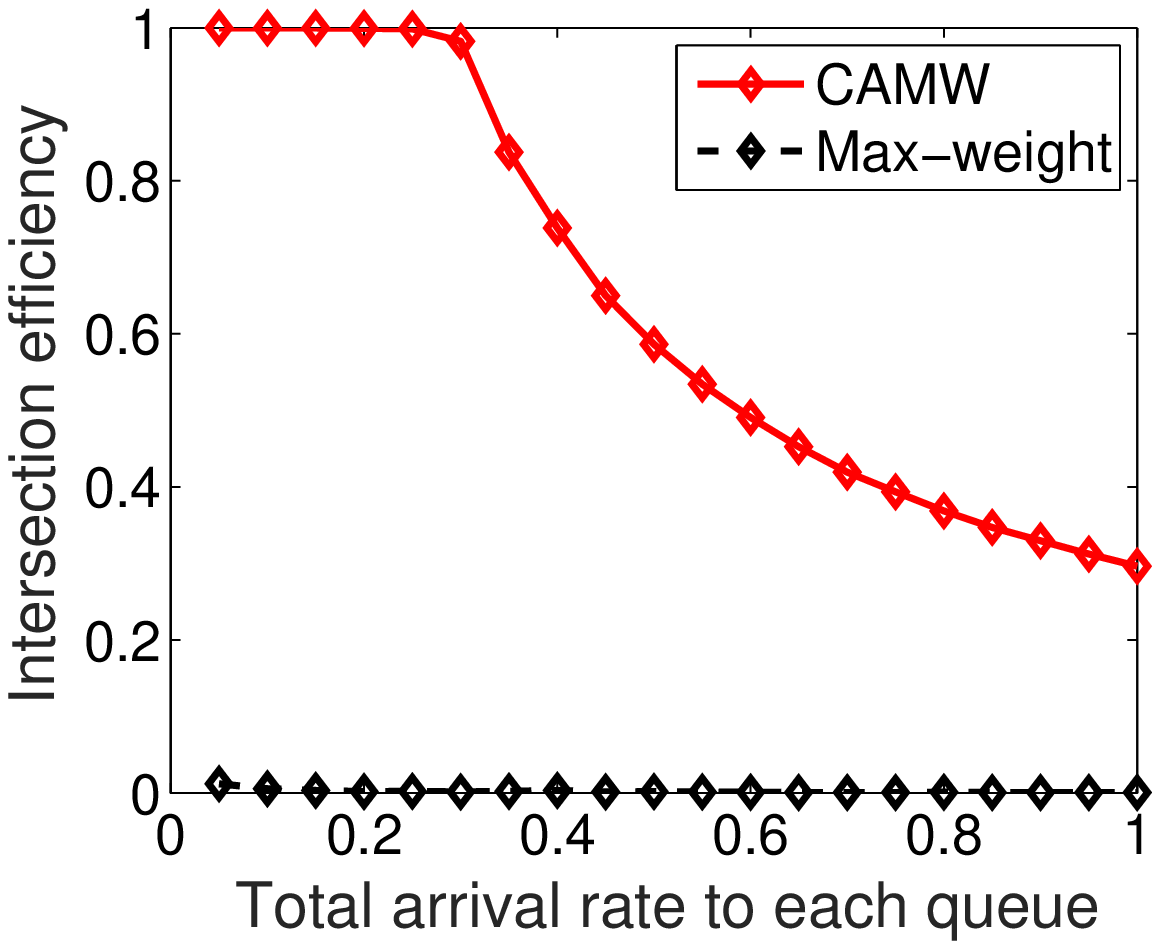}} } \\
\subfigure[$\rho=0.7$]{ \scalebox{.30}{\includegraphics{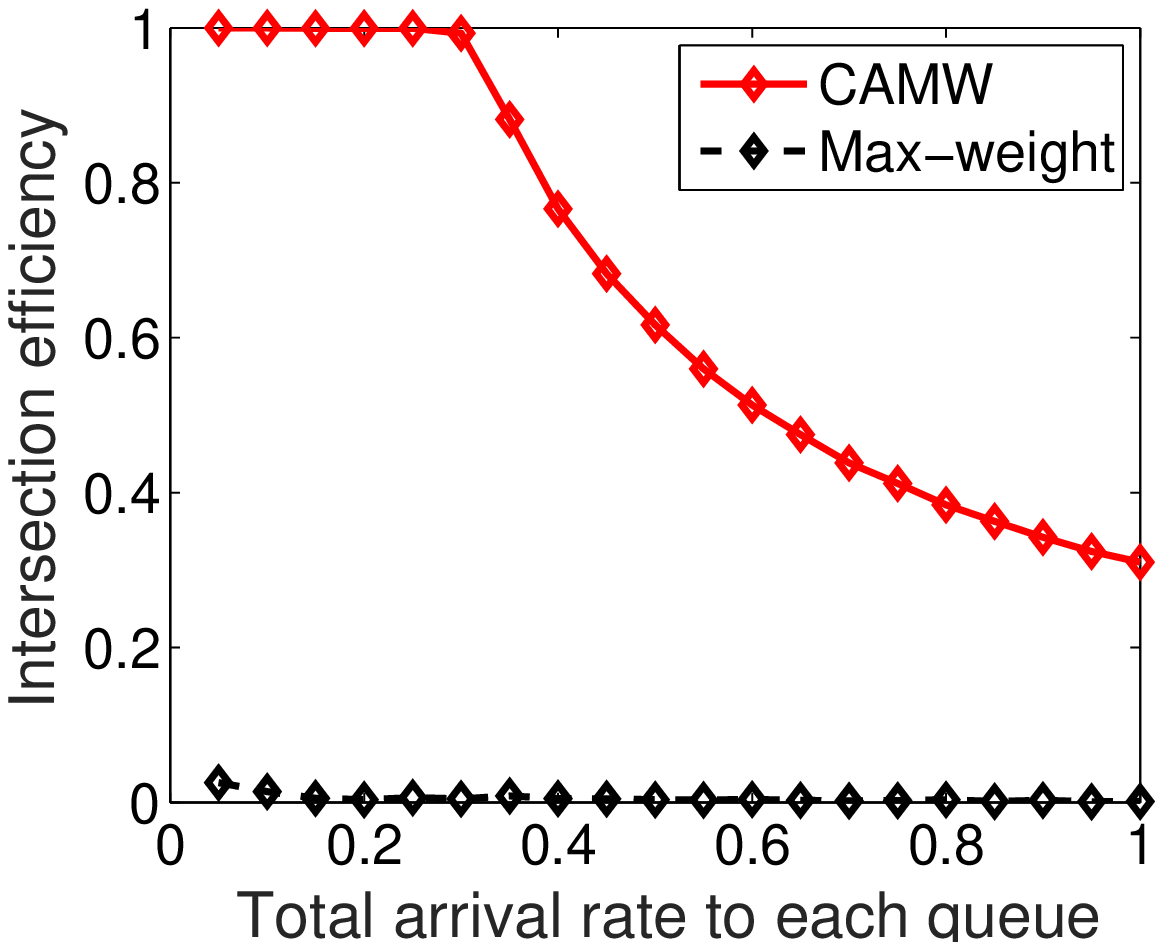}} }  
\subfigure[$\rho=1.0$]{ \scalebox{.30}{\includegraphics{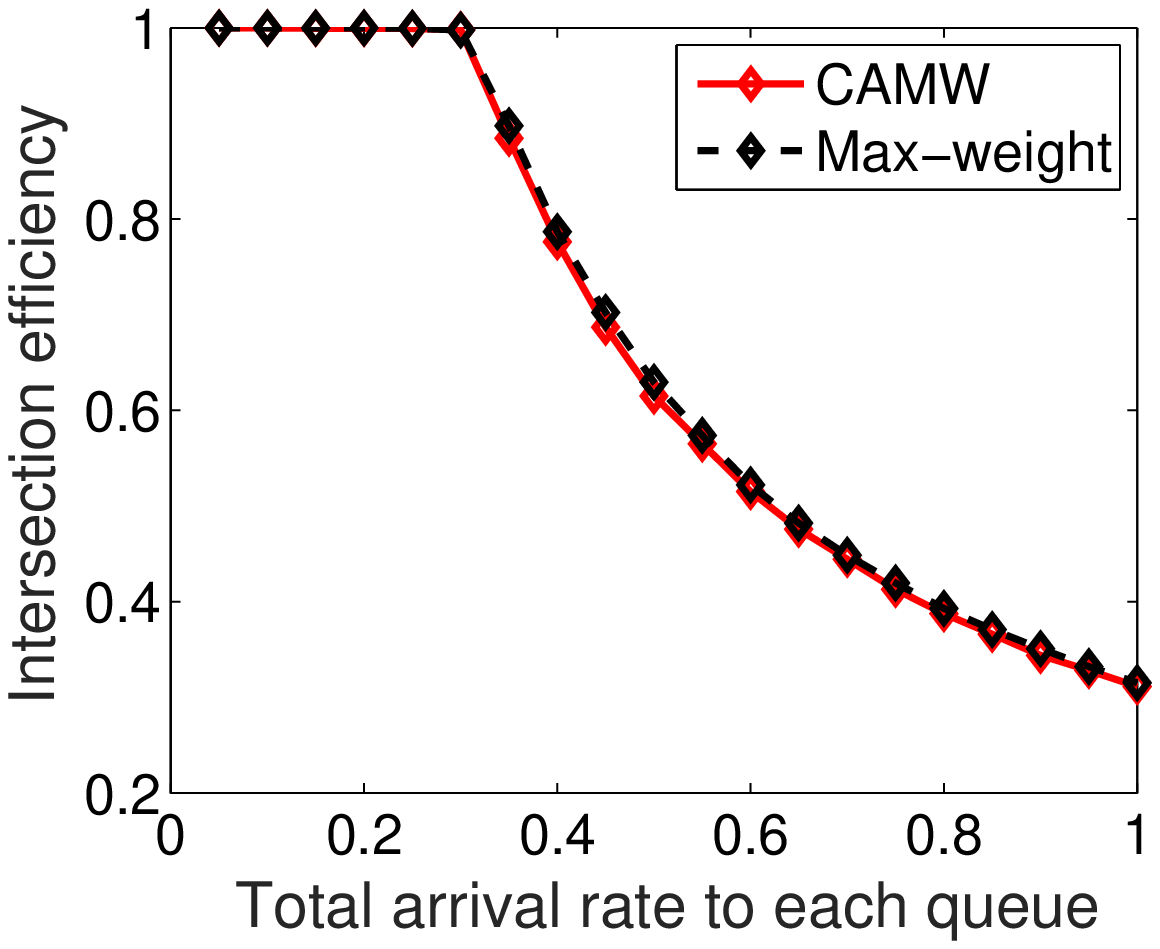}} }
\end{center}
\begin{center}
\vspace{-20pt}
\caption{\label{fig:throuputs} Intersection efficiency versus total arrival rate to each queue with different communication probability $\rho$ for \modelI~. Each green phase lasts for two time slots and each queue has the same arrival rate and $\lambda_1=1.5\lambda_2$.}
\vspace{-30pt}
\end{center}
\end{figure}

Fig. \ref{fig:throuputs} presents the intersection efficiency versus total arrival rate to each queue for different communication probability $\rho$. The intersection efficiency is defined as the ratio of departing traffic to arrival traffic. In this setup, each green phase lasts for two time slots. Each queue has the same arrival rate, and $\lambda_1=1.5\lambda_2$. It can be observed that when $\rho=1.0$, both of the algorithms can achieve very similar intersection efficiency. However, if $\rho \neq 1$, the intersection efficiency of max-weight scheduling algorithm drops almost to zero, while CAMW can still achieve satisfying intersection efficiency thanks to taking into account heterogeneous communication probabilities. 

%

\subsection{Evaluation of CAMW for \modelII}
In this section, we assume all the queues in the intersection follow \modelII. The evolution of the average queue size of the intersection using CAMW and max-weight algorithm for different communication probability $\rho$ is presented in Fig. \ref{fig:2hol_queuesize_evol}. The arrival rate to each queue is $\lambda_1=\lambda_2=0.2$ and each green phase lasts for two time slots. It can be observed from Fig. \ref{fig:2hol_queuesize_evol}(a) that when communicating probability $\rho$ is small, CAMW is slightly better than the max-weight algorithm, which is because both of the two algorithm select traffic phases in a similar way when $\rho$ is small. The average queue sizes over 10,000 time slots  when $\rho=0.1$ using max-weight and CAMW are 10.6601 and 9.0236, respectively. It can be observed from Fig. \ref{fig:2hol_queuesize_evol}(b) that when communicating probability $\rho$ is large, our algorithm improves much over max-weight algorithm. This is because the estimation accuracy in our algorithm improves as $\rho$ increases, which allows more vehicles to pass at each green phase. When $\rho=0.9$, the average queue size over 10,000 time slots using max-weight and CAMW is 10.6601 and 4.3873, respectively. Note that CAMW performs better than max-weight when $\rho$ increases in \modelII, which is against the observation we had in \modelI. The reason is that while $\rho$ affects max-weight's decision about HoL vehicles as explained in (\ref{eq:sche_obj_max_weight}) in \modelI, it does not have any effect in \modelII. 


\begin{figure}[t!]
\begin{center}
\subfigure[$\rho=0.1$]{ \scalebox{.30}{\includegraphics{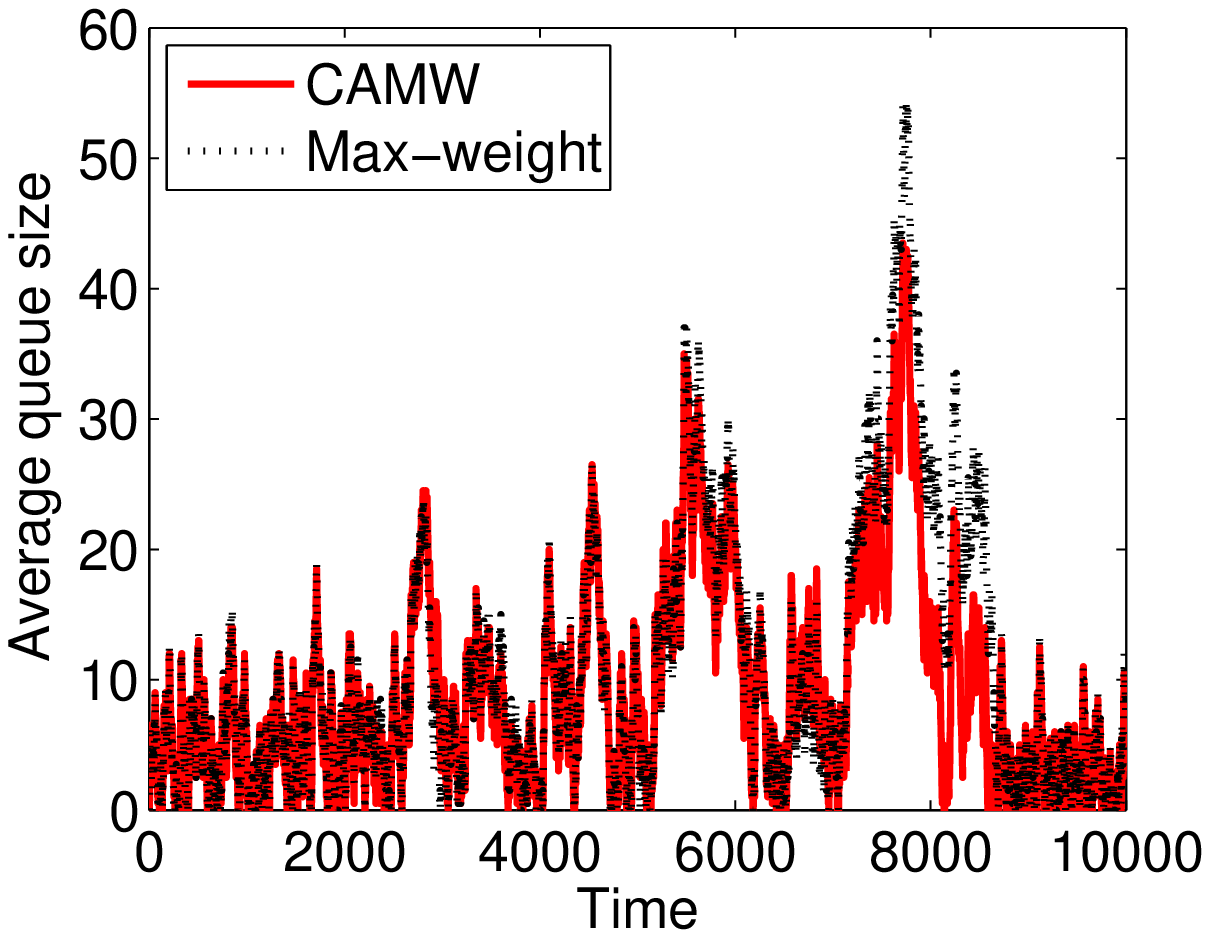}} } 
\subfigure[$\rho=0.9$]{ \scalebox{.30}{\includegraphics{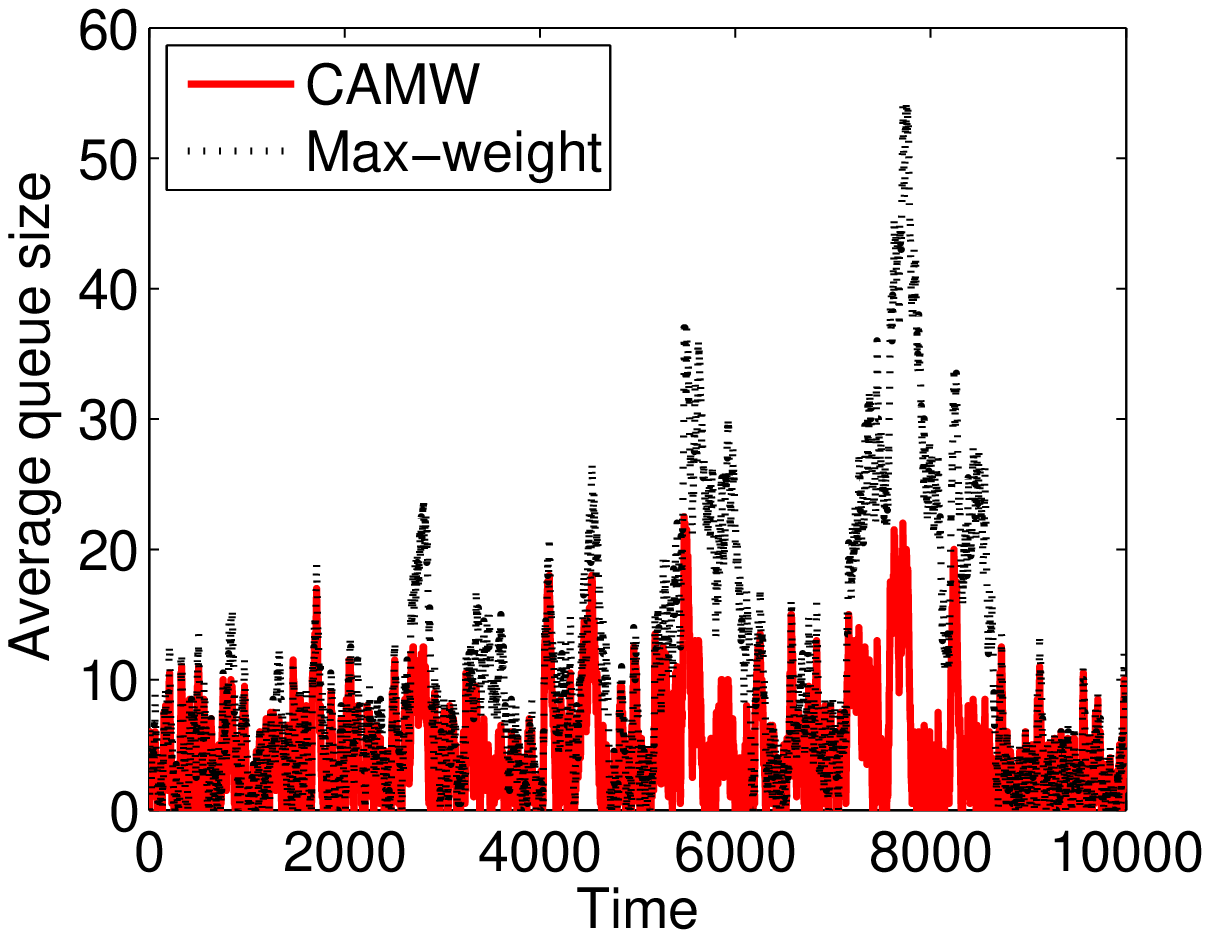}} } \\
\end{center}
\begin{center}
\vspace{-20pt}
\caption{\label{fig:2hol_queuesize_evol} The evolution of the average queue size of the intersection using our algorithm and max-weight algorithm for different communication probability $\rho$ for \modelII. The arrival rate to each queue is $\lambda_1=\lambda_2=0.2$ and each green phase lasts for two time slots.}
\vspace{-30pt}
\end{center}
\end{figure}

Fig. \ref{fig:2hol_thro_vs_arrival} presents the intersection efficiency versus total arrival rate to each queue for different communication probabilities $\rho$. Each queue has the same total arrival rate and $\lambda_1=\lambda_2$. Each green phase lasts for two time slots. It can be observed that the performance of our algorithm improves as the communicating probability $\rho$ increases, while max-weight has the same performance as $\rho$ changes. The reason is that the estimation accuracy in our algorithm improves as $\rho$ increases, so CAMW performs better than the max-weight algorithm as $\rho$ increases. Note that CAMW improves over max-weight by $14\%$, which is significant. 

\begin{figure}[t!]
\begin{center}
\subfigure[$\rho=0.1$]{ \scalebox{.30}{\includegraphics{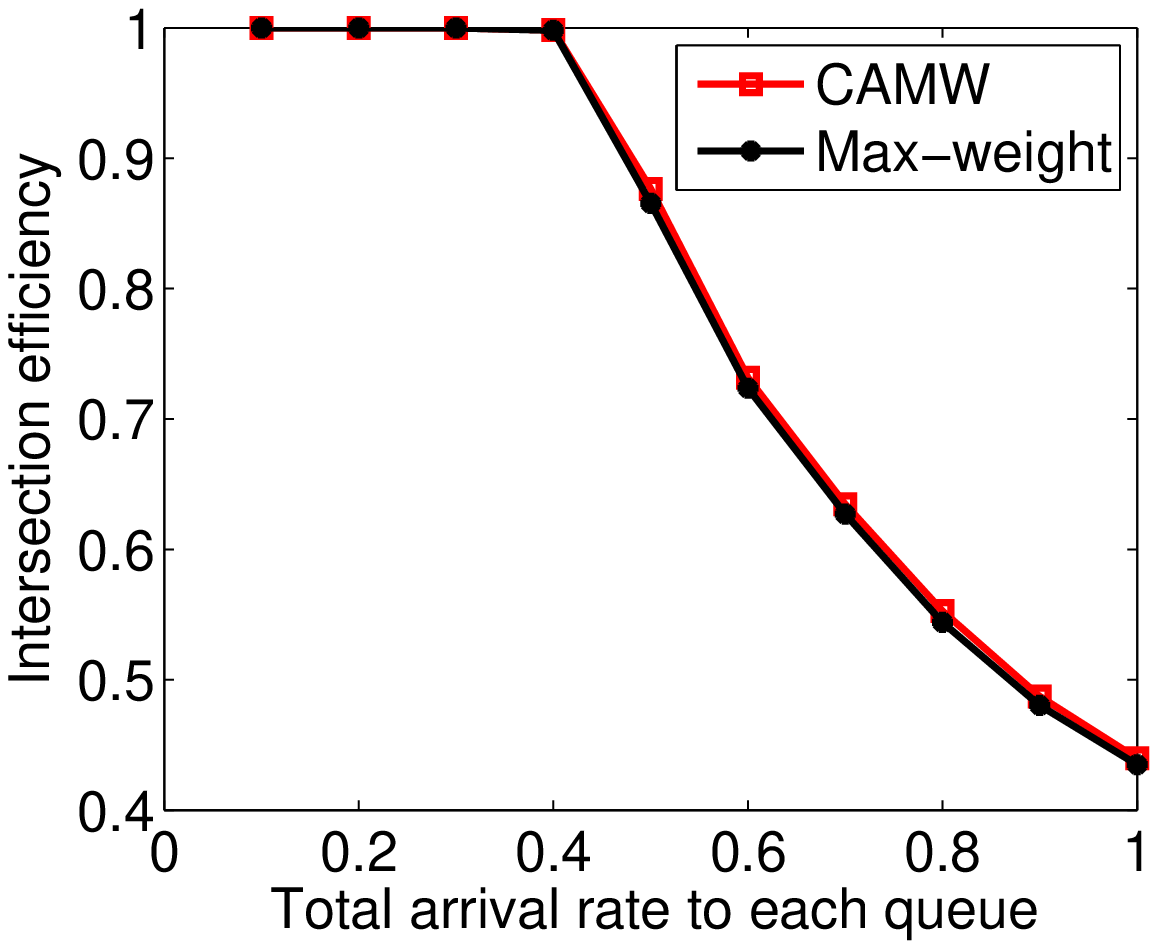}} } 
\subfigure[$\rho=0.4$]{ \scalebox{.30}{\includegraphics{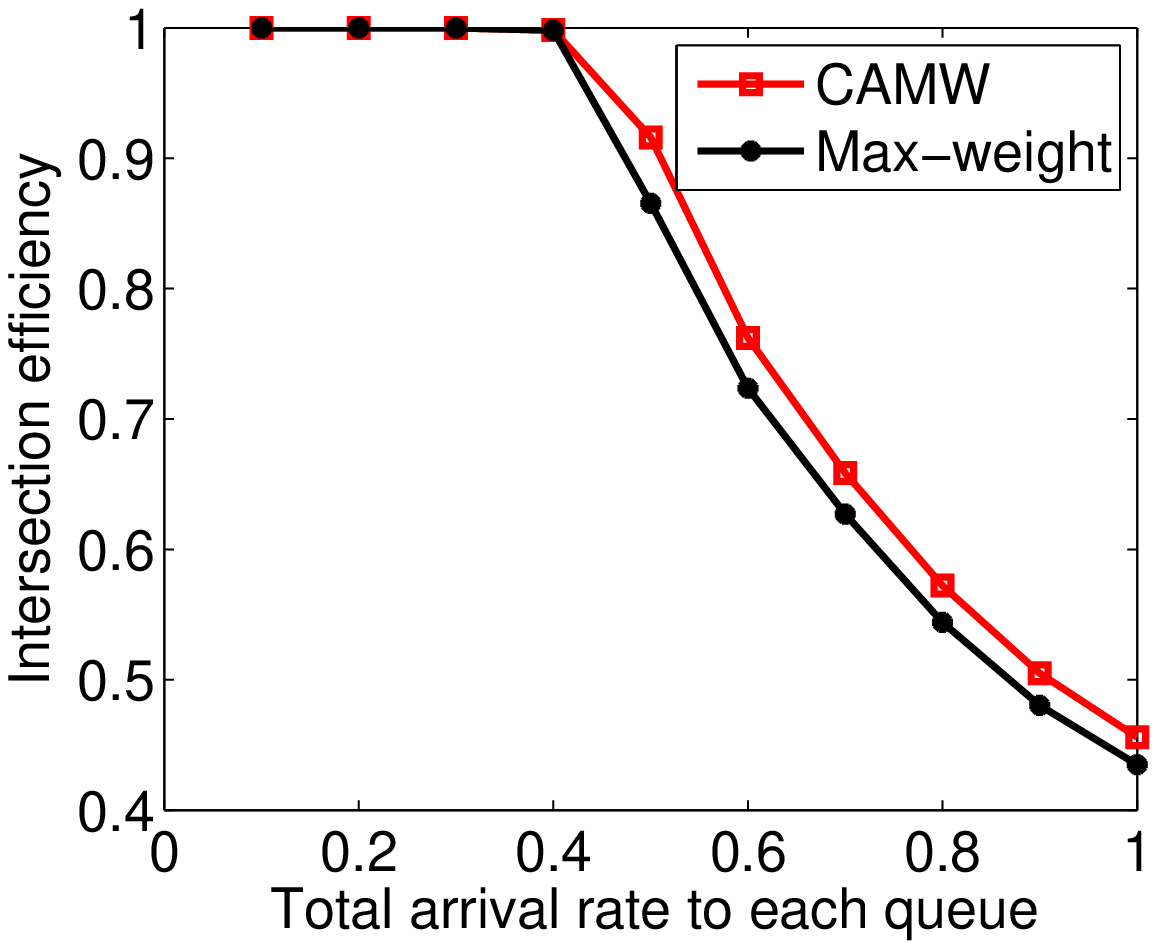}} } \\
\subfigure[$\rho=0.7$]{ \scalebox{.30}{\includegraphics{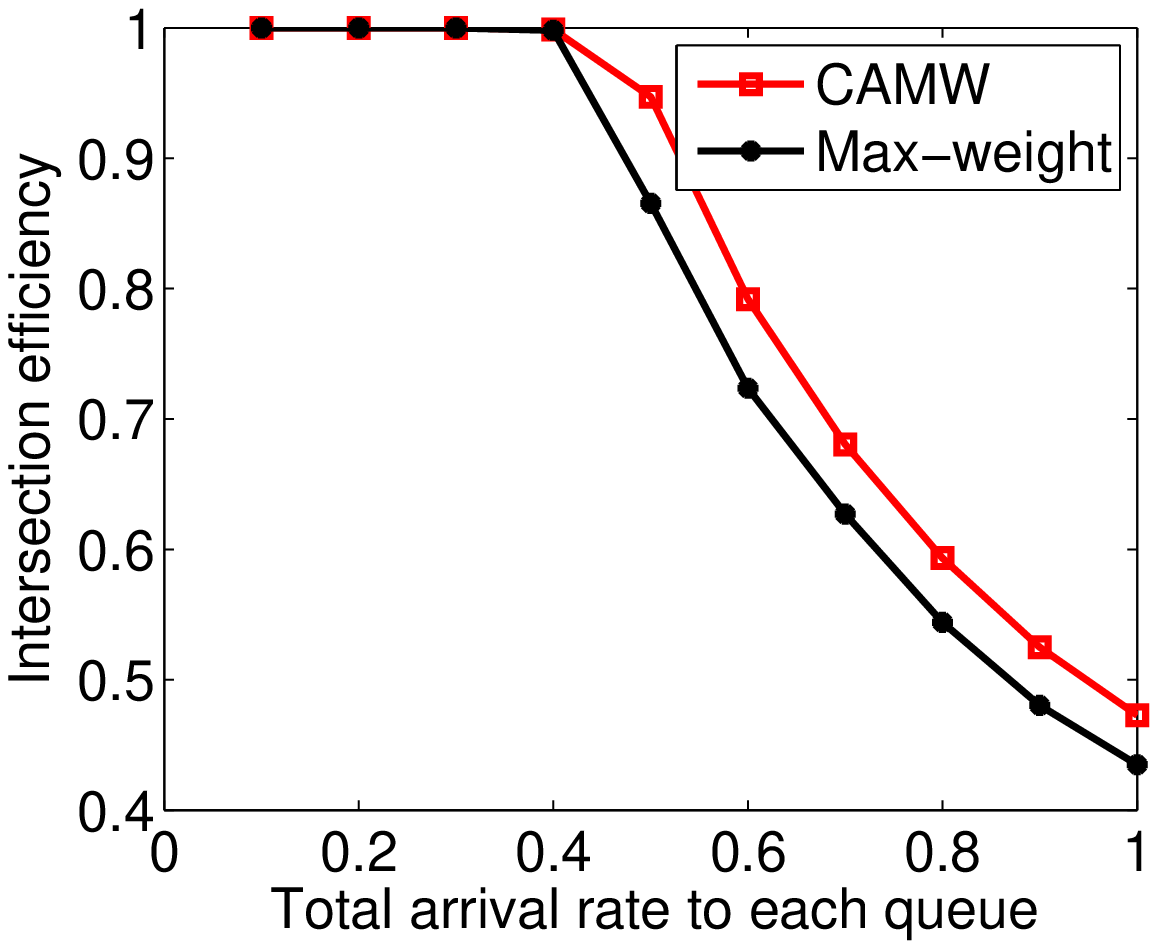}} }  
\subfigure[$\rho=1.0$]{ \scalebox{.30}{\includegraphics{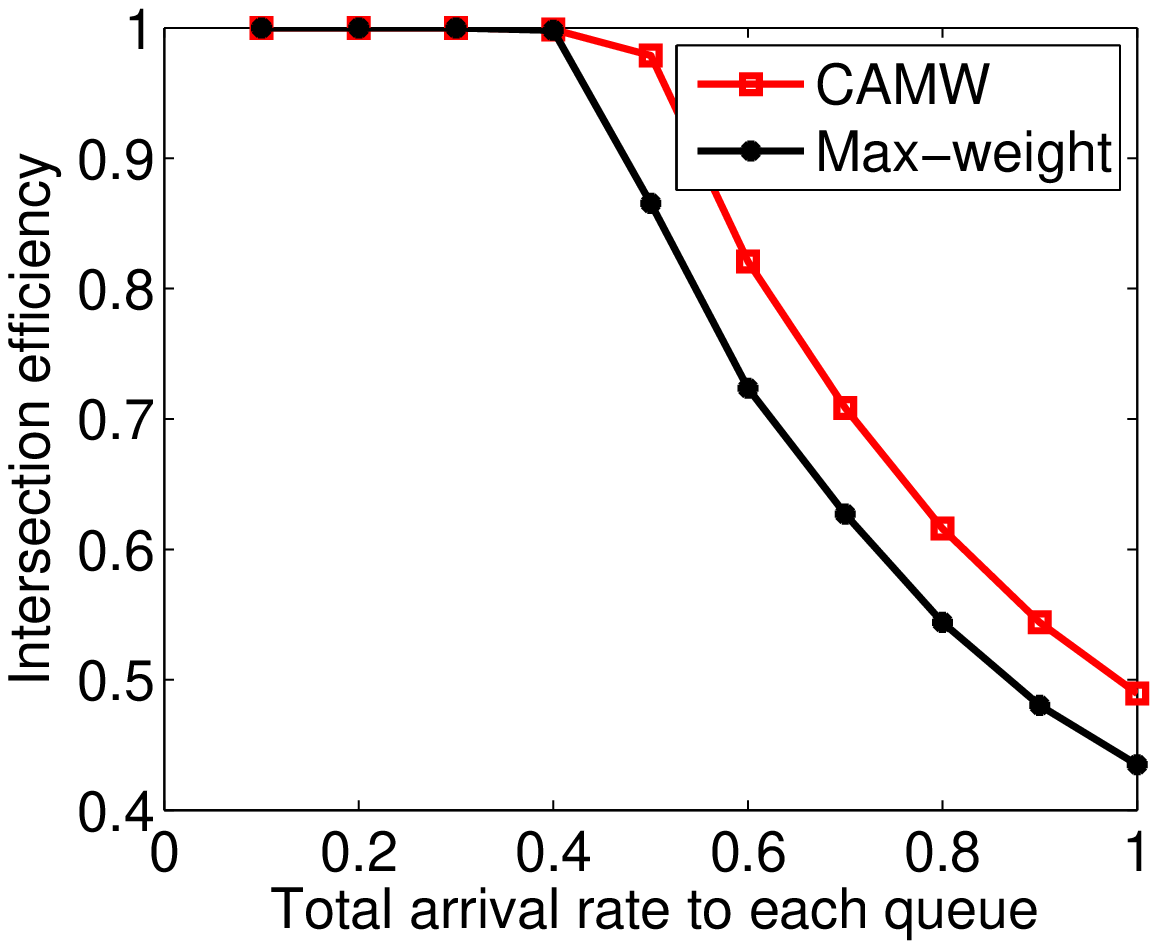}} }
\end{center}
\begin{center}
\vspace{-20pt}
\caption{\label{fig:2hol_thro_vs_arrival} Intersection efficiency versus total arrival rate to each queue with different communication probability $\rho$ for \modelII. Each queue has the same total arrival rate and $\lambda_1=\lambda_2$, and each green phase lasts for two time slots.}
\vspace{-35pt}
\end{center}
\end{figure}

\vspace{-5pt}
\section{Conclusion}\label{sec:conclusion}
In this paper, we considered a transportation system of heterogeneously connected vehicles, where not all vehicles are able to communicate. For this setup, we developed a connectivity-aware max-weight scheduling (CAMW) algorithm by taking into account the connectivity of vehicles. The crucial components of CAMW are expectation and learning components, which determine the estimated number of vehicles that can pass through the intersections by taking into account the heterogeneous communications. The simulations results show that CAMW algorithm significantly improves the intersection efficiency over max-weight. 

\bibliographystyle{abbrv}
\vspace{-5pt}
\bibliography{sigproc}  
%
%
\appendix
\vspace{-5pt}
\section{Proof of Theorem 1} \label{proof:k(t)_m1}
In this section, we specifically focus on the calculation of $E(K^{\phi=1}_i(t))$, where $\phi=1$ corresponds to the phase in Fig.~\ref{fig:phases}(a) to explain our the proof in an easier way. Note that $E(K^{\phi=1}_i(t))$ calculation can be directly generalized to $E(K^{\phi}_i(t))$, $\forall \phi \in \Phi$.

We first derive the calculation of $E(K^{\phi=1}_i(t))$ when all communicating vehicles are going straight. The calculation of $E(K^{\phi=1}_i(t))$ for other cases will be obtained based on this derivation. If all communicating vehicles are going straight at time slot $t$, we can consider the queue as divided into $(T+1)$ {\em blocks} by the $T$ communicating vehicles. (Note that $T$ is the number of communicating vehicles in a queue). 

Let a random variable $J$ denote the number of vehicles that can pass the intersection. The probability that $j$ vehicles pass the intersection is $P[J=j]$, and it behaves similarly to the geometric distribution. However, the probability distribution is different when $j$ falls into different {\em blocks} due to the communicating vehicles that go straight. To be more precise, we have 

\begin{align}
P[J=j] & = \left\{\begin{array}{rl}p^j_1p_2,& 1 \leq j \leq v_1-2 \\ 
p^{j-1}_1p_2 , & v_1 \leq j \leq v_2-2 \\
\vdots\\
p^{j-T}_1p_2 , & v_T \leq j \leq n-1 \\
p^{n-T}_1 , & j=n \\
\end{array}\right.
\label{eqn:distribution_allstraight}
\end{align}
Note that $P[J=v_1-1]$, $P[J=v_2-1]$, $\cdots$, $P[J=v_T-1]$ are all 0. The reason is that the communicating vehicles at locations $v_1, v_2, \cdots, v_T$ are all going straight, and if $v_l-1$ vehicles can pass the intersection. Then, $v_l$ vehicles can pass the intersection for sure ($l=1,2,\cdots,T$). 

Using (\ref{eqn:distribution_allstraight}), we can obtain the expected number of vehicles that can pass the intersection as $E(K^{\phi=1}_i(t))$ when all communicating vehicles are going straight. \Ie
\begin{eqnarray}
E(K^{\phi=1}_i(t))=\sum_{j=1}^{v_1-2}jp^j_1p_2+\sum_{j=v_1}^{v_2-2}jp^{j-1}_1p_2+\cdots \nonumber \\
+\sum_{j=v_T}^{n-1}jp^{j-T}_1p_2+np^{n-T}_1
\label{eqn:estimation_k(t)_allstraight}
\end{eqnarray}
In (\ref{eqn:estimation_k(t)_allstraight}), $\sum_{j=v_l}^{v_{l+1}-2}jp^{j-l}_1p_2$ can be expressed as $p^{1-l}_1$ $p_2$ $\sum_{j=v_l}^{v_{l+1}-2}$ $jp^{j-1}_1=p^{1-l}_1p_2\frac{\partial(\sum_{j=v_l}^{v_{l+1}-2}p^j_1)}{\partial p_1}=\frac{p^{1-l}_1}{p_2}((p_1+p_2v_l)p^{v_l-1}_1+(1-2p_1-p_2v_{l+1})p^{v_{l+1}-2}_1)$. Thus, we can obtain $E(K^{\phi=1}_i(t))$ when all communicating vehicles are going straight as
\begin{eqnarray}
E(K^{\phi=1}_i(t))  =  \sum_{l=0}^{T}\frac{p^{1-l}_1}{p_2}((p_1+p_2v_l)p^{v_l-1}_1 \nonumber \\
 +  (1-2p_1-p_2v_{l+1})p^{v_{l+1}-2}_1)+np^{n-T}_1
\label{eqn:compact_k(t)_allstr8}
\end{eqnarray}
Note that we have $v_0=1, v_{T+1}=n+1$ in (\ref{eqn:compact_k(t)_allstr8}) to make it consistent with (\ref{eqn:estimation_k(t)_allstraight}). 

When there are some communicating vehicles going left, let $v_L(t)$ be the location of the first communicating vehicle that goes left. There are $(L-1)$ communicating vehicles in front of $v_L(t)$ that going straight and $(T-L)$ communicating vehicles behind $v_L(t)$ which will be blocked for sure. Now, we only focus on the vehicles between the location 1 to $(v_L(t)-1)$. There are $(L-1)$ communicating vehicles among them, and all of the communicating vehicles are going straight. Thus, we can use the similar analysis as used in (\ref{eqn:estimation_k(t)_allstraight}) except that now the maximum number of vehicles that can pass the intersection is $(v_L(t)-1)$ instead of $n$. Therefore, we have the expected number of vehicles that can pass the intersection $E(K^{\phi=1}_i(t))$ when the first communicating vehicle that turns left is at location $v_L(t)$. Thus, 
\begin{eqnarray}
E(K^{\phi=1}_i(t))  =  \sum_{l=0}^{L-1}\frac{p^{1-l}_1}{p_2}((p_1+p_2v_l)p^{v_l-1}_1  +  \nonumber \\
(1-2p_1-p_2v_{l+1})p^{v_{l+1}-2}_1)+(v_L(t)-1)p^{v_L(t)-L}_1
\end{eqnarray}

By taking into account all the $(T+1)$ situations, we conclude that 
\begin{align}
E(K^{\phi}_i(t)) = \begin{cases} \sum_{l=0}^{T(t)}\frac{p^{1-l}_1}{p_2}((p_1+p_2v_l(t))p^{v_l(t)-1}_1\nonumber\\
 + (1-2p_1-p_2v_{l+1}(t))p^{v_{l+1}(t)-2}_1)\nonumber\\
+np^{n-T(t)}_1, & \mbox{    { if} $C_1$  {holds}  }\\
\\
\sum_{l=0}^{L-1}\frac{p^{1-l}_1}{p_2}((p_1+p_2v_l(t))p^{v_l(t)-1}_1\nonumber \\
 + (1-2p_1-p_2v_{l+1}(t))p^{v_{l+1}(t)-2}_1)\nonumber \\
 +(v_L(t)-1)p^{v_L(t)-L}_1,   & \mbox{    { if} $C_2$  {holds.}  }
\end{cases} 
\\
\end{align}

By following the same analysis, we can obtain $E(K^{\phi}_i(t))$ for $\phi=2,3,4$. This concludes the proof.
\end{document}